\documentclass[journal,onecolumn]{IEEEtran}

\usepackage[utf8]{inputenc}
\usepackage[T1]{fontenc}
\usepackage{url}
\usepackage{ifthen}
\usepackage{cite}
\usepackage[cmex10]{amsmath}
\usepackage{amssymb}
\usepackage{amsmath}
\usepackage{amsthm}
\usepackage{etoolbox}
\usepackage{amsfonts}
\usepackage{array,multirow,graphicx}
\usepackage{float}
\usepackage{tikz}
\usepackage{algpseudocode}
\usepackage{algorithm}
\usepackage{textcomp}
\usepackage{makecell}
\usepackage{relsize}
\usepackage{mathrsfs}
\algnewcommand\algorithmicforeach{\textbf{for each}}
\algdef{S}[FOR]{ForEach}[1]{\algorithmicforeach\ #1\ \algorithmicdo}

\newcommand{\STAB}[1]{\begin{tabular}{@{}c@{}}#1\end{tabular}}
\newtheorem{proposition}{Proposition}[section]

\usetikzlibrary{arrows}
\usetikzlibrary{decorations.pathreplacing,angles,quotes}
\usetikzlibrary{shapes.geometric, arrows}
\tikzstyle{arrow} = [thick,->,>=stealth]
\newtheorem{theorem}{Theorem}[section]

\makeatletter
\patchcmd{\@maketitle}
  {\addvspace{0.5\baselineskip}\egroup}
  {\addvspace{-0.9\baselineskip}\egroup}
  {}
  {}
\makeatother

\begin{document}
\title{Cooperative Network Coding for Distributed Storage using Base Stations with Link Constraints}



\author{\IEEEauthorblockN{Suayb S. Arslan\IEEEauthorrefmark{1},
Massoud Pourmandi\IEEEauthorrefmark{2} and
Elif Haytaoglu\IEEEauthorrefmark{3}}

\IEEEauthorblockA{\IEEEauthorrefmark{1} Dept. of Computer Engineering,
MEF University, Maslak, Istanbul, TR} \\
\IEEEauthorblockA{\IEEEauthorrefmark{2} Dept. of EE, Boğaziçi University University, Bebek, Istanbul, TR} \\
\IEEEauthorblockA{\IEEEauthorrefmark{3} Dept. of Computer Engineering, Pamukkale University, Denizli, TR}
\thanks{This study is supported by TUBITAK under grant no 119E235.}}


\maketitle

\begin{abstract}
In this work,  we consider a  novel distributed data storage/caching scenario in a  cellular setting where multiple nodes may fail/depart at the same time. In order to maintain the target reliability, we allow cooperative regeneration of lost nodes with the help of base stations allocated in a set of hierarchical layers.  Due to this layered structure,  a  symbol download from each base station has a different cost, while the link capacities connecting the nodes of the cellular system and the base stations are also limited. In this more practical and general scenario,  we present the fundamental trade-off between repair bandwidth cost and the storage space per node.  Particularly interesting operating points are the minimum storage as well as bandwidth cost points in this trade-off curve. We provide closed-form expressions for the corresponding bandwidth  (cost)  and storage space per node for these operating points.  Finally,  we provide an explicit optimal code construction for the minimum storage regeneration point for a given set of system parameters.
\end{abstract}


\section{Introduction}

Scalability and reliability of distributed data storage and caching are quite crucial for next generation cellular system design. In such settings, nodes that store information may fail or leave the cell permanently. For fault tolerance, the data is usually cached in erasure coded form in the cell. Moreover, in order to maintain the target reliability (or service quality in a constant distributed network), we need to regenerate the unavailable content on a periodic basis. A data of size $F$ symbols is erasure coded to provide required durability and distributed across different $n$ mobile nodes of the local cell where each node stores $\alpha$ units of data. In the case of a failure or node departure, either an instantaneous repair process regenerate lost node contents as soon as they become unavailable or  a lazy repair may be adapted in which no regeneration take place during a time duration to accumulate losses and initiate a cooperative regeneration. The latter helps reduce the required bandwidth dramatically \cite{Kermarrec2011},\cite{Shum2013}.

In any lazy cooperative repair scheme, there are two phases. Suppose that at the time of regeneration there are $t$ failed/departed nodes accumulated so far. Newcomer nodes are regenerated by first contacting $d \leq  n - t$ other live nodes to download $\beta$ units of data to resurrect the lost data content. In the second phase, after each newcomer processes the downloaded data, the $t$ newcomers exchange data among themselves, each by downloading $\beta^\prime$ units of data from  the rest of the other $t-1$ newcomers.

In this paper, we consider the availability of multiple base stations which cache the stored content in the nodes of the local cellular network. In particular, we introduce  two novel concepts, 
namely \textit{bandwidth cost} and \textit{link capacity} in addition to cooperation. More specifically, we constrain the maximum  link capacity limit to be $b_l \beta$ to be downloaded from the $l$-th base station. The bandwidth cost defines the cost we pay per downloaded symbols from each base station. Although there are similarities to previous rack-aware configurations \cite{sohn2019,hou2019,Tebbi2019,zhang2021rackaware, Hou2020}, the essence of the problem is more general as it encompasses cost notion and capacity-constrained communication channels. Moreover, these configurations do not address the cooperative regeneration and/or minimizing the bandwidth cost in their rack-aware architecture.

The rest of the paper is organized as follows. In Section II, we provide our motivation with an example coding scenario in which base station assistance is beneficial, and then we introduce  the problem statement formally. In Section III, we provide our main result that characterizes the trade-off between the bandwidth cost versus storage space per cellular node. We explore two important operating points in the trade-off curve; namely base station assisted minimum storage and minimum bandwidth cost cooperative regeneration (namely BS-MSCR and BS-MBCCR, respectively).  In Section IV, we provide a code construction for BS-MSCR operating point based on  MDS codes. Finally, we conclude our paper in Section V.

\section{Motivation and System Model}

\subsection{Motivating Example}
Let us consider a recurring example in the literature whereby out of four storage nodes, two nodes fail and we attempt cooperative regeneration of data in these nodes. In that setting, each node stores two packets where each packet is assumed to be 1MB ($F=4$ MB). The first node stores $A_1,A_2$, second node stores $B_1,B_2$, third node contains $A_1+B_1$ and $2A_2+B_2$. Finally, the fourth node stores $2A_1+B_1$ and $A_2+B_2$. Assume further that the second and fourth nodes fail. To recover both nodes the classical way and simultaneously, they need to contact the remaining alive nodes and download four packets per lost node (4MB) for the regeneration process. However, studies such as \cite{Kermarrec2011} and \cite{Shum2013} showed that cooperation can help reduce the amount of data exchange to 3MB per node. As can be seen in Fig. \ref{Ex1}, an example introduced in \cite{Shum2013}, the top newcomer downloads $A_1$ and $A_1+B_1$ and computes $B_1$ for local storage and $2A_1+B_1$ to transfer to other newcomer. Similarly, the bottom newcomer downloads $A_2$ and $2A_2+B_2$ from alive nodes and compute $A_2+B_2$ for local storage and $B_2$ to transfer to other newcomer.

\begin{figure}[t!]
\includegraphics[width=0.4\columnwidth, height=2.3in]{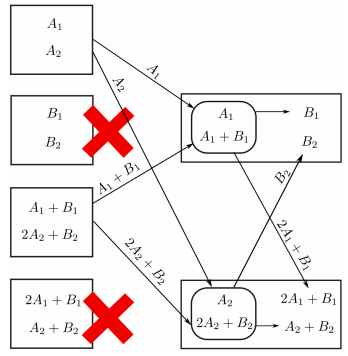}
\centering
\caption{ Cooperative regeneration of multiple failures \cite{Shum2013}.} \label{Ex1}
\end{figure}

Let us assume a two-layer scenario in which different base stations are involved in the recovery process. We partition packets to half packets (0.5MB) and assume BS communication to be $w_1>1$ more costly compared to local communication. Also, communication with the satellite/cloud layer costs $w_2>w_1>1$ whereby BS is able to provide at most $r$ fraction of local downloaded data ($\beta$) and satellite/cloud layer stores the entire file. 
This constraint on the link capacity implies that the number of bytes that can be communicated with higher layer nodes is a linear function of that of local nodes ($\beta$)\footnote{This constraint will be a parameter of our work later on.}. 
In that case, one packet transfer between local nodes and higher layers correspond to $w_1$ or $w_2$ packet transfers solely between local nodes subject to this limitation. 

We provide an example scenario in Fig. \ref{Ex2} in which half packets for $A_1$ namely, $A_1^1$ and $A_1^2$ and for $A_2$ namely, $A_2^1$ and $A_2^2$ are summed before transmission to newcomers. The top newcomer initially resurrect $B_1^1+B_1^2$ and downloads from the other newcomer $B_2^1+B_2^2$. Finally by downloading $B_1^2$ from the satellite/cloud layer  and $B_2^2$ from the BS, it would be able to complete the repair process. Similarly, the bottom newcomer initially computes $A_2^1+B_2^1+A_2^2+B_2^2$, and through downloading $2A_1^1+B_1^1+2A_1^2+B_1^2$ from the other newcomer and the half packets $2A_1^2+B_1^2$ from the satellite/cloud layer  and $A_2^2+B_2^2$ from the BS, respectively, it would be able to complete the repair process. Overall a total bandwidth cost of $3+w_1+w_2$ 
is experienced for the exact repair. In Table \ref{Table1}, we compare the bandwidth cost per repaired node  using this coding scheme (CoopLayer) based on the example in \cite{Shum2013} (CoopLocal) with $r = 1$. 
We have also included no cooperation (NoCoopLocal) and a protocol which only contacts higher layers and download the node content directly from them without local assistance (FullLayer). As can be seen, the coding scheme we present downloads $0.5$MB from the BS and $0.5$MB from the satellite/cloud layer corresponding to $r=1$ fraction of the local download ($\beta = 0.5$MB) and performs best among other schemes. In this study, we will argue that the presented coding scheme uses both layers and is optimal for $r=1$ for the weight values of Scenario 1 given in Table \ref{Table1}. In this case, the cost of CoopLayer is the sum of $1$ unit cost  from the helper nodes, $0.5$ unit cost from the cooperation phase, $0.55$ from the BS and $0.85$ from the satellite layers. However it turns out that in Scenario 2, using different weight distribution, it is optimal to use only the BS in the repair process. In other words, using $r = 1$ and smaller packetizations, we can show that there exists an optimal coding method that can achieve a bandwidth cost of as low as  $43/15 \approx 2.867$ (corresponds to $8/3 \approx 2.66$ MBs). 

\begin{figure}[t!]
\includegraphics[width=0.7\columnwidth, height=3in]{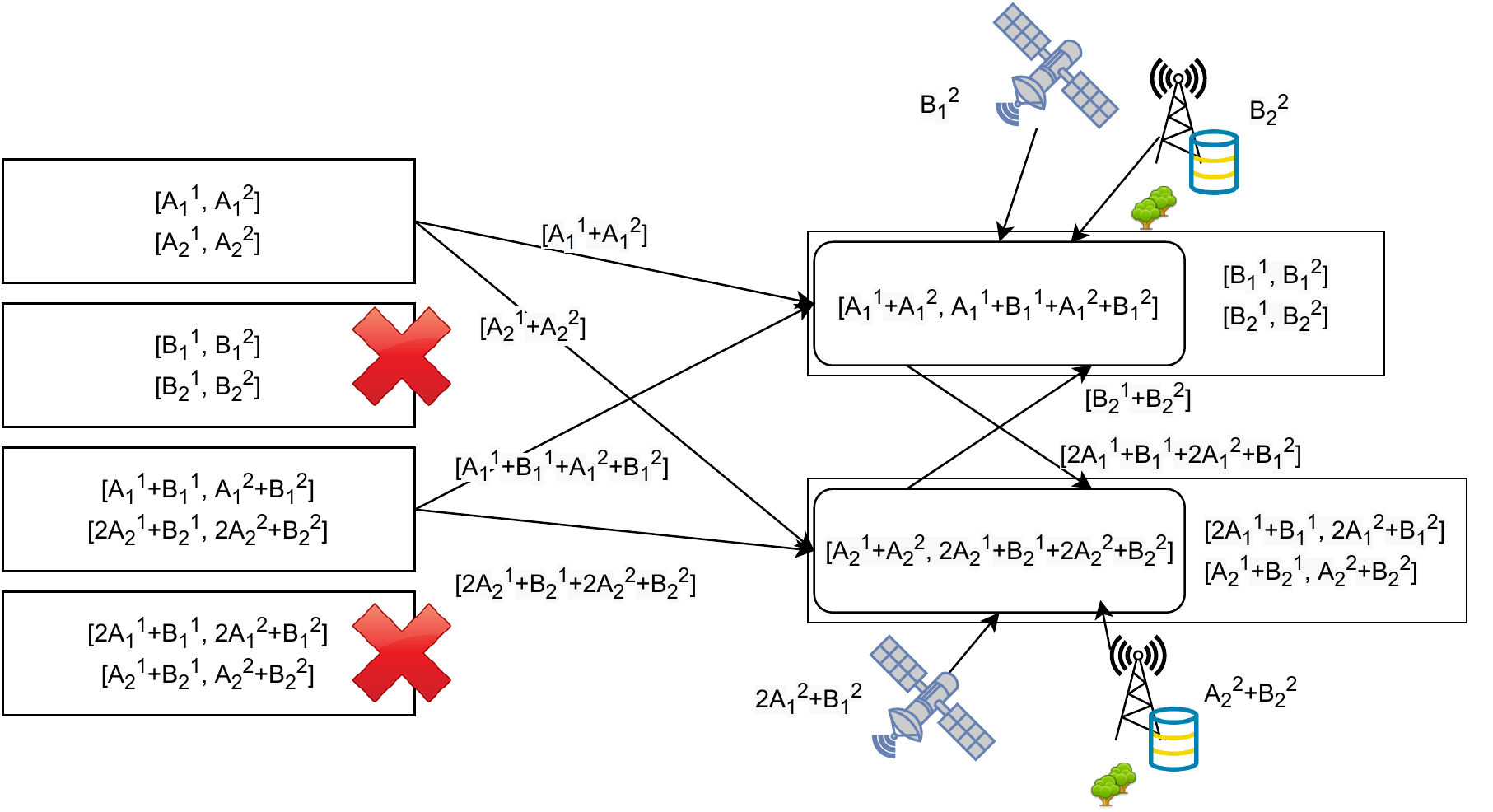}
\centering
\caption{Cooperative regeneration of multiple failed nodes with the help of a base station and a satellite/cloud layer.} \label{Ex2}
\vspace{-5mm}
\end{figure}

\begin{table}[b!]
  \centering
  \begin{tabular}{|c|l|l|l|r|r|r|r|}
    \Xhline{4\arrayrulewidth}
     &  \multicolumn{1}{c|}{$\beta$} & \multicolumn{1}{c|}{ $w_1$} & \multicolumn{1}{c|}{$w_2$} & \multicolumn{1}{c|}{\STAB{\rotatebox[origin=c]{90}{NoCoopLocal}}} & \multicolumn{1}{c|}{\STAB{\rotatebox[origin=c]{90}{CoopLocal}}} & \multicolumn{1}{c|}{\STAB{\rotatebox[origin=c]{90}{CoopLayer}}} & \multicolumn{1}{c|}{\STAB{\rotatebox[origin=c]{90}{FullLayer}}} \\ \Xhline{4\arrayrulewidth}
          \multirow{1}{*}{\STAB{\rotatebox[origin=c]{0}{Scenario 1}}}
      & 0.5 MB & 1.1  & 1.7                    &   4                        &       3                    &         2.9                  &          3.1                  \\ \hline
     \multirow{1}{*}{\STAB{\rotatebox[origin=c]{0}{Scenario 2}}}
      & 2/3 MB & 1.3  & 2.1                   &    4                       &       3                    &         2.867                  &          3.667               \\ \Xhline{4\arrayrulewidth}
  \end{tabular}
  \vspace{1mm}
  \caption{Bandwidth cost per newcomer for cooperative, local and BS-assisted cooperative regeneration schemes ($r=1$).}
\label{Table1}
\end{table}

\subsection{System Model and Information Flow Diagram}

We consider a local cell with $n$ storage nodes and multiple upper layers serving as back-up cluster nodes, each with size $n_l$  with  $l=1,..,M$, such that individual clusters are assigned a cost factor $w_l$. We suppose that clusters are stacked according to non-descending costs, therefore cluster with label 0 would correspond to local cluster and cluster with label $M$ is the cluster that has highest cost factor. We further assume each link to have limited capacity $b_l$ and the number of symbols that can be downloaded (${\beta_l}^{\prime\prime}$)  is constrained i.e., ${\beta_l}^{\prime\prime } \leq b_l \beta$.

In our study we invoke information flow graphs to characterize the trade-off region, which are graphical representations of a network in which some information sources are  multicast to a set of destination sinks through intermediate nodes. An information flow graph can be described by three types of nodes namely \textit{source nodes}, \textit{intermediate nodes} and \textit{sink nodes}. In general, for the case of single source multiple sink nodes problem, achievable capacity of the network is characterized through max-flow min-cut  theorem \cite{Ahlswede2000},\cite{Dimakis2010}. In general case  of multiple source nodes, achievable capacity  is still an open problem.  As a pioneering work, Alswede \textit{et al.} \cite{Ahlswede2000} characterized the min-cut max-flow  with network coding and proved that linear network codes are sufficient to achieve the min-cut bound in multicast problem with a single source. In the context of a distributed storage system, a directed acyclic graph called an information flow graph $\mathcal{G}$ is used. As mentioned before, $\mathcal{G}$ consists of a source $\mathbf{S}$, intermediary
nodes, and data collectors ($\mathbf{DC}$) which only contact any $k$ devices for the full reconstruction of the stored file.

We characterize our proposed information flow graph with $(n,k,d,t,\alpha,\gamma,\{ {r_l}\},\{ {w_l}\})$ tuple using $M$ distinct layers, different than the one in \cite{Shum2013}. An example graph is illustrated in Fig. \ref{figflow} for $M=2$ case. Besides, for a given code parameters $(n,k,d,t,\alpha,\gamma)$ , we denote all possible graphs with $\mathcal{G}(n,k,d,t;\alpha,\beta,\beta^\prime,\{ {\beta _l}^{\prime \prime }\} _{l = 1}^M)$. Data file $\mathcal{F}$ with size $F$ is represented by the source node $\mathbf{S}$ at stage $s=-1$. Initial nodes are represented by pairs of vertices $(x_{I{n_j}}^i,x_{Ou{t_j}}^i)$ at stage $s=0$, where \(i\) and \(j\) are the time index and device label introduced at time index  \(i\), respectively. Secondary nodes are represented by tuple $(x_{I{n_j}}^i,x_{Coo{p_{j,1}}}^i,x_{Coo{p_{j,2}}}^i,x_{Ou{t_j}}^i)$ with stage $s \ge 1$. Finally, data collectors are represented by single node with stage $s \ge 1$.


\begin{figure}[t!]
  \centering 
  \resizebox{6.3in}{4in}{%
\begin{tikzpicture}[line cap=round,line join=round,>=triangle 45,x=1.2cm,y=1.2cm]
\draw [line width=0.75pt] (-7,3) circle (0.3cm);
\draw [line width=0.75pt] (0.5,2) circle (0.3cm);
\draw [line width=0.75pt] (-6,4) circle (0.2cm);
\draw [line width=0.75pt] (-5,5) circle (0.25cm);
\draw [line width=0.75pt] (-5,6) circle (0.25cm);
\draw [line width=0.75pt] (-6,3) circle (0.2cm);
\draw [line width=0.75pt] (-6,2) circle (0.2cm);
\draw [line width=0.75pt] (-6,1) circle (0.2cm);
\draw [line width=0.75pt] (-5,4) circle (0.2cm);
\draw [line width=0.75pt] (-5,3) circle (0.2cm);
\draw [line width=0.75pt] (-5,2) circle (0.2cm);
\draw [line width=0.75pt] (-5,1) circle (0.2cm);
\draw (-7,3) node[anchor=center] {$\mathbf{S}$};
\draw (0.5,2) node[scale=0.7,anchor=center] {$\mathbf{DC}$};
\draw (-6,4) node[scale=0.5,anchor=center] {${I{n_1}}$};
\draw (-5,4) node[scale=0.5,anchor=center] {$Ou{t_1}$};
\draw (-5,5) node[scale=0.5,anchor=center] {$B{S_1}$};
\draw (-5,6) node[scale=0.5,anchor=center] {$B{S_2}$};
\draw (-6,3) node[scale=0.5,anchor=center] {$I{n_2}$};
\draw (-5,3) node[scale=0.5,anchor=center] {$Ou{t_2}$};
\draw (-6,2) node[scale=0.5,anchor=center] {$I{n_3}$};
\draw (-5,2) node[scale=0.5,anchor=center] {$Ou{t_3}$};
\draw (-6,1) node[scale=0.5,anchor=center] {$I{n_4}$};
\draw (-5,1) node[scale=0.5,anchor=center] {$Ou{t_4}$};
\draw (-3.5,4) node[scale=0.5,anchor=center] {$I{n_5}$};
\draw [line width=0.75pt] (-3.5,4) circle (0.2cm);
\draw (-2.5,4) node[scale=0.4,anchor=center] {$Coo{p_{5,1}}$};
\draw [line width=0.75pt] (-2.5,4) circle (0.26cm);
\draw (-1.5,4) node[scale=0.4,anchor=center] {$Coo{p_{5,2}}$};
\draw [line width=0.75pt] (-1.5,4) circle (0.26cm);
\draw (-0.5,4) node[scale=0.5,anchor=center] {$Ou{t_5}$};
\draw [line width=0.75pt] (-0.5,4) circle (0.2cm);
\draw (-3.5,3) node[scale=0.5,anchor=center] {$I{n_6}$};
\draw [line width=0.75pt] (-3.5,3) circle (0.2cm);
\draw (-2.5,3) node[scale=0.4,anchor=center] {$Coo{p_{6,1}}$};
\draw [line width=0.75pt] (-2.5,3) circle (0.26cm);
\draw (-1.5,3) node[scale=0.4,anchor=center] {$Coo{p_{6,2}}$};
\draw [line width=0.75pt] (-1.5,3) circle (0.26cm);
\draw (-0.5,3) node[scale=0.5,anchor=center] {$Ou{t_6}$};
\draw [line width=0.75pt] (-0.5,3) circle (0.2cm);
\draw [arrow] (-6.822256648517049,3.2416760248837444) -- (-6.17606064562056,3.9051177094307294);
\draw [arrow] (-6.7003013757649095,3.0134437580904336) -- (-6.199630684830645,3.0121486490626728);
\draw [arrow] (-6.754544033688956,2.827514149559126) -- (-6.184618297688359,2.0769160851750312);
\draw [arrow] (-6.935940014457439,2.7069192632527908) -- (-6.175365243042706,1.0961614867427316);
\draw [arrow] (-5.800026612852102,3.996737419334807) -- (-5.199973851719383,3.996766019710285);
\draw [arrow] (-5.80001078991864,3.0020774626431517) -- (-5.1996960421907845,2.98897771650989);
\draw [arrow] (-5.8000016443658895,2.0008110139653263) -- (-5.199998576855749,2.000754490341379);
\draw [arrow] (-5.800433831213405,1.0131660273370575) -- (-5.199837951327249,0.9919505770810372);
\draw [arrow] (-4.91,2.18) -- (-3.69,3.99);
\draw [arrow] (-4.8,2.01) -- (-3.69,2.94);
\draw [arrow] (-4.96,1.2) -- (-3.58,3.82);
\draw [arrow] (-4.81,1.07) -- (-3.55,2.82);
\draw [arrow] (-4.75,5.04) ..controls (-3.66,4.71)and(-3,4).. (-2.58,3.21);
\draw [arrow] (-4.75,5.94) ..controls (-3.74,5.33)and(-2.86,4).. (-2.45,3.23);
\draw [arrow] (-4.75,5.94) ..controls (-3.67,5.53)and(-2.82,4.88).. (-2.54,4.23);
\draw [arrow] (-4.75,5.02) ..controls (-3.88,4.87)and(-3.16,4.57).. (-2.69,4.14);
\draw (-3.53,5.5) node[scale=0.6,anchor=center] {${w_2r_2}\beta$};
\draw (-3.53,5.08) node[scale=0.6,anchor=center] {${w_2r_2}\beta$};
\draw (-3.9,4.94) node[scale=0.6,anchor=center] {${w_1r_1}\beta$};
\draw (-3.9,4.45) node[scale=0.6,anchor=center] {${w_1r_1}\beta$};
\draw (-4.06,3.57) node[scale=0.6,anchor=center] {$\beta$};
\draw (-4.06,3.1) node[scale=0.6,anchor=center] {$\beta$};
\draw (-4.47,2.39) node[scale=0.6,anchor=center] {$\beta$};
\draw (-4.47,1.7) node[scale=0.6,anchor=center] {$\beta$};

\draw [arrow] (-7,3.3) -- (-5.22,5.89);
\draw [arrow] (-5,5.75) -- (-5,5.25);
\draw [arrow] (-3.3,4) -- (-2.74,4);
\draw [arrow] (-2.21,4) -- (-1.74,4);
\draw [arrow] (-1.21,4) -- (-0.7,4);

\draw [arrow] (-3.3,3) -- (-2.74,3);
\draw [arrow] (-2.21,3) -- (-1.74,3);
\draw [arrow] (-1.21,3) -- (-0.7,3);
\draw [arrow] (-2.31,3.16) -- (-1.61,3.78);
\draw [arrow] (-2.31,3.8) -- (-1.63,3.19);
\draw [arrow] (-0.32,2.91) -- (0.3,2.2);
\draw [arrow] (-4.81,1.07) -- (0.2,1.94);
\draw[densely dotted](-8,4.63) -- (1,4.63);
\draw[densely dotted](-8,5.63) -- (1,5.63);
\draw (-7.5,5.8) node[scale=0.6, anchor=center] {Layer $2$};
\draw (-7.5,4.8) node[scale=0.6, anchor=center] {Layer $1$};
\draw (-7.5,3.5) node[scale=0.6, anchor=center] {Layer $0$};
\draw[decoration={brace,mirror,raise=5pt},decorate]
  (-6.68,0.63) -- node[below=6pt] {\tiny Stage $-1$} (-5.54,0.63);
\draw[decoration={brace,mirror,raise=5pt},decorate]
  (-5.45,0.63) -- node[below=6pt] {\tiny Stage $0$} (-3,0.63);
  \draw[decoration={brace,mirror,raise=5pt},decorate]
  (-2.91,0.63) -- node[below=6pt] {\tiny Stage $1$} (0,0.63);
\draw (-6.6,3.61) node[anchor=center] {${\mathsmaller\alpha}$};
\draw (-6.6,3.09) node[anchor=center] {${\mathsmaller\alpha}$};
\draw (-6.6,2.8) node[anchor=center] {${\mathsmaller\alpha}$};
\draw (-6.6,2.26) node[anchor=center] {${\mathsmaller\alpha}$};
\draw (-6.6,2.26) node[anchor=center] {${\mathsmaller\alpha}$};
\draw (-5.65,4.09) node[scale=0.5,anchor=center] {$+ \infty$};
\draw (-5.65,3.09) node[scale=0.5,anchor=center] {$+ \infty$};
\draw (-5.65,2.09) node[scale=0.5,anchor=center] {$+ \infty$};
\draw (-5.65,1.09) node[scale=0.5,anchor=center] {$+ \infty$};
\draw (0.14,2.64) node[scale=0.5,anchor=center] {$+ \infty$};
\draw (-1,1.85) node[scale=0.5,anchor=center] {$+ \infty$};
\draw (-6.14,4.78) node[scale=0.7,anchor=center] {$F$};
\draw (-4.82,5.59) node[scale=0.5,anchor=center] {$F_{1}$};
\draw (-3.14,4.09) node[scale=0.5,anchor=center] {$+ \infty$};
\draw (-3.14,3.09) node[scale=0.5,anchor=center] {$+ \infty$};
\draw (-2.07,4.09) node[scale=0.5,anchor=center] {$+ \infty$};
\draw (-2.07,3.09) node[scale=0.5,anchor=center] {$+ \infty$};
\draw (-1.14,4.09) node[scale=0.5,anchor=center] {$\alpha$};
\draw (-1.14,3.09) node[scale=0.5,anchor=center] {$\alpha$};
\draw (-2.21,3.5) node[scale=0.5,anchor=center] {$\beta '$};
\draw (-2.21,3.34) node[scale=0.5,anchor=center] {$\beta '$};

\filldraw[red] (-4.89,4)--(-4.68,4.18)--(-4.75,4.26)--(-4.95,4.08)--cycle;
\filldraw[red] (-4.94,4.19)--(-4.86,4.25)--(-4.71,4.07)--(-4.78,3.99)--cycle;
\filldraw[red] (-4.89,3)--(-4.68,3.18)--(-4.75,3.26)--(-4.95,3.08)--cycle;
\filldraw[red] (-4.94,3.19)--(-4.86,3.25)--(-4.71,3.07)--(-4.78,2.99)--cycle;
\end{tikzpicture}
}
\caption{An example information flow diagram for $M=2$, $n_1=n_2=1$. For short-hand notation, only subscripts are used. For instance, $x_{I{n_1}}^0, x_{Coo{p_{5,1}}}^1$ are denoted by  $I{n_1}, Coo{p_{5,1}}$.} \label{figflow}
\vspace{-6mm}
\end{figure}
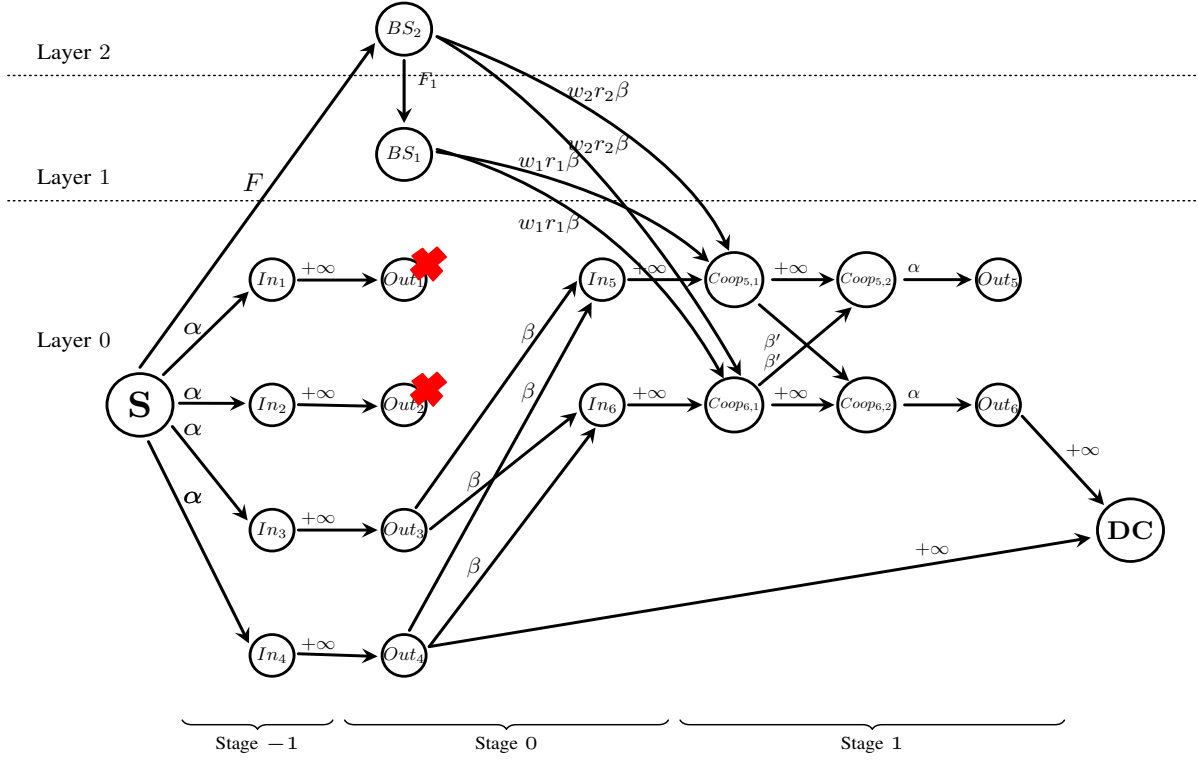

In our repair configuration, the initial data repair stage involves contacting $d$ other local live nodes and downloading $\beta$ symbols from each through edges $x_{Ou{t_j}}^k \to x_{I{n_{j'}}}^i$ for $k < i$ and $j \not= j'$.
In the next phase, the newcomer contacts $M$ above layers (all $n_l$ nodes in them) and download $\beta^{\prime\prime}_l$ symbols from the $l$-th layer whereby each layer only stores $F_l \geq b_l \beta$ symbols for some $b_l\in\mathbb{R}$ and $j$-th node in $l$-th layer is associated with cost $w_{l} \geq 1$ per symbol download. In the final phase, all newcomers undergo a joint local coordination by downloading $\beta^{\prime}$ symbols from at most  $t-1$ other failed nodes within the same local cell through the edges $x_{Coo{p_{j,1}}}^i \to x_{Coo{p_{j',2}}}^i$ for $j \not= j^\prime$. The newcomer stores $\alpha$ symbols, and this is represented
by a directed edge from $x_{Coo{p_{j,2}}}^i$ to $x_{Ou{t_j}}^i$ with capacity $\alpha$. To successfully reconstruct the file, a data collector is connected to $k$ alive “out” nodes with
distinct indices, but not necessarily from the same stage, by $k$
infinite-capacity edges. 
We  finally note  that  all  previous  works would be a special case of this general description.

\subsection{Optimization Problem}

Without loss of generality, we shall assume each layer contains a single base station. Moreover, similar to previous works \cite{Calis2016}, we introduce auxiliary variables $r_l$ to be able to express $\beta_l^{\prime\prime}$ in terms of $\beta$ i.e., $\beta^{\prime\prime}_l =  r_l\beta$. We characterize the link constraints by bounding the $r_l$ i.e. $0 \leq r_l \leq b_l$ for some fixed $b_l \in \mathbb{R}$. Based on the regeneration description, the total repair bandwidth cost per  failed node (${\gamma_c}$) is obtained as follows,
\begin{equation}\label{gamma_c_general}
  {\gamma _c}(\mathbf{s}) = d\beta  + (t - 1)\beta^{'} + \sum\limits_{l = 1}^M {{s_l}{w_l}{r_l}\beta }
\end{equation}
where for some $\rho \in \{0,1,\dots, M\}$ and entries $s_l$, which defines the number of used BSs with
\begin{equation}\label{s_def}
  \mathbf{s}_\rho = \left\{ {(\underbrace {1,...,1}_\rho,\underbrace {0,...,0}_{M - \rho}):\sum\limits_{l = 1}^M {{s_l}}  = \rho} \right\}.
\end{equation}

The minimum cut of the flow diagram $G$ induces the following constraint on the file size
\begin{align}\label{thrm1_1}
\mathop {\min }\limits_{r_l, \mathbf{u} \in P} \Bigg\{
{u_0}\alpha  + \sum\limits_{j = 1}^g {{u_j}(d^\prime - \sum\limits_{i = 0}^{j - 1} {{u_i}} )\beta }
 + {u_j}\left( {t - {u_j}} \right)\beta^{\prime} \Bigg\} \geq F
\end{align}
where $P = \{\mathbf{u} = (u_i)_{0 \leq i < g}:1\leq u_i\leq t$ and $\sum_i^g u_i = k\}$ and $d' = d + \sum\limits_{l = 1}^M s_l r_l$. Here, $u_i$ is the number of nodes contacted in each repair group
of size $t$ during the repair and $g$ is the number of repairing stages. Note that unlike $d$, $d^\prime$ is real valued. Therefore, slightly modifying the steps given in \cite{Shum2013} we can obtain the file size constraint in \eqref{thrm1_1}, which can further be expressed for $g=0,...,k$ as,
\begin{equation} \label{thrm1_2}
F \leq
\begin{cases} 
      \Phi_{g,\mathbf{s}_\rho} + g(\beta/2-\beta^\prime) & \text{if  $\beta \geq 2 \beta^\prime$}, \\
      \Phi_{g,\mathbf{s}_\rho} + \psi_{g,t}(\beta/2-\beta^\prime) & \text{otherwise}
   \end{cases}
\end{equation}
where $\Phi_{g,\mathbf{s}} = \alpha (k - g) + g\beta (d' - k + \frac{{g}}{2}) + \beta^\prime gt$ used as a shorthand notation and $\psi_{g,t} = \lfloor g/t\rfloor t^2 + (g - \lfloor g/t\rfloor t)^2$.




For a given file size $F$ and system parameters $\alpha, d, k, t, \mathbf{w}= (w_l)_{1\leq l \leq M}, \mathbf{b}=(b_l)_{1\leq l \leq M}$, the bandwidth cost-storage trade-off is the solution to the following constrained optimization problem
\begin{equation}\label{optim_prob}
  \mathop {\min }\limits_{\mathbf{s},{r_1},...,{r_M},\beta ,\beta '} \left\{ {d\beta  + (t - 1)\beta ' + \sum\limits_{l = 1}^M {{s_l}{w_l}{r_l}} \beta } \right\}
\end{equation}
subject to the conditional constraint in \eqref{thrm1_2} and
\begin{align}\label{constraint_optim}
&0 \le \beta ,\beta ' \le \frac{F}{d}\\
& {{b_l} \in \mathbb{R}, {r_l} \in [0,b_l] }, \quad l = 1,...,M\\
&\mathbf{s}_\rho = [{s_1},...,{s_M}],{s_l} \in \{ 0,1\}
\end{align}
Due to space constraints, we present a solution of the optimization problem using a few selected parameters in Fig. \ref{fig_k6d9t3}. As can be seen the new trade-off curve, which uses different number of base stations at different operating points, allows better achievable region as compared to two major past works. 

\begin{figure}[t!]
\includegraphics[width=0.7\columnwidth,height=3in]{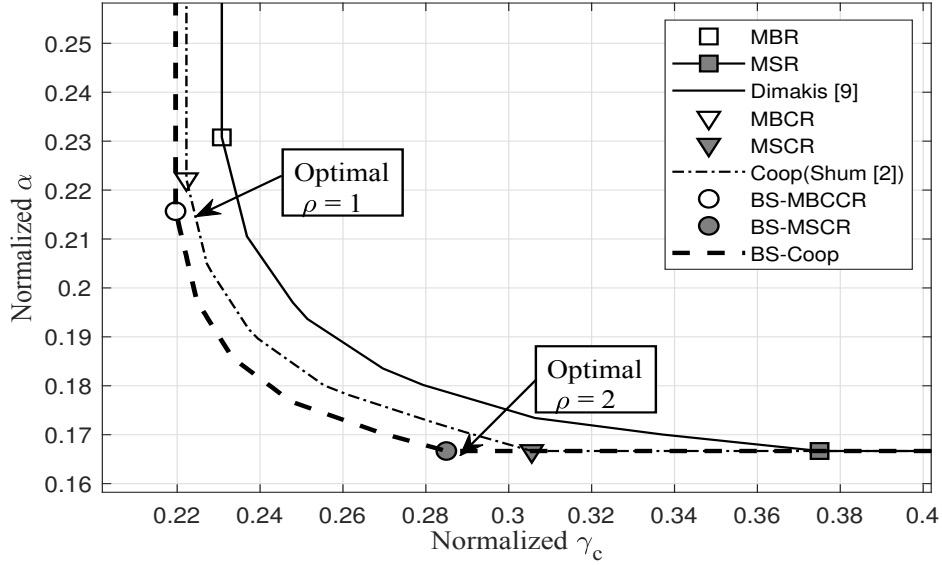}
\centering
\caption{Storage versus repair bandwidth cost trade off,for $k=6$, $d=9$, $t=3$, $b=[1,0.75,0.5,0.25] $ and $w=[1.2,1.4,1.8,1.84]$.} \label{fig_k6d9t3}
\vspace{-4mm}
\end{figure}

\section{BS-MSCR and BS-MBCCR operating points}

Two operating points of this trade-off curve are of special importance namely, minimum storage and bandwidth cost operating points. 
\begin{theorem} \label{thm31}
For a given set values of $r_1,...,r_M$ in a BS-assisted cooperative repair scenario, the minimum storage and minimum repair bandwidth cost regeneration, namely BS-MSCR and BS-MBCCR points on the tradeoff curve are characterized as
\begin{align}\label{bs_points}
&\left( {{\gamma _{\textrm{BS-MSCR}}},{\alpha _{\textrm{BS-MSCR}}}} \right) = 
\left( {\frac{{F(d + \sum\limits_{l = 1}^M {{s_i}{w_i}{r_i}}  + t - 1)}}{{k(d + \sum\limits_{l = 1}^M {{s_i}{r_i}}  + t - k)}},\frac{F}{k}} \right)\\
&\left( {{\gamma _{\textrm{BS-MBCCR}}},{\alpha _{\textrm{BS-MBCCR}}}} \right) =\nonumber\\
&\left( {\frac{{F(2(d + \sum\limits_{l = 1}^M {{s_i}{w_i}{r_i}} ) + t - 1)}}{{k(2(d + \sum\limits_{l = 1}^M {{s_i}{r_i}} ) + t - k)}},\frac{{F(2(d + \sum\limits_{l = 1}^M {{s_i}{r_i}} ) + t - 1)}}{{k(2(d + \sum\limits_{l = 1}^M {{s_i}{r_i}} ) + t - k)}}} \right)
\end{align}
\end{theorem}

\begin{proof}
At the BS-MBCCR operating point, similar to the reasoning given in \cite{Shum2013}, for a given set of values of $r_1,...,r_M$, the minimum repair bandwidth cost is the intersection between $\beta=2\beta^\prime$ and $\beta-\beta^\prime$ planes corresponding to the constraints given in \eqref{thrm1_2} for $g=k$ and $\alpha  \ge d\beta  + (t - 1)\beta^{\prime} + \sum\limits_{l = 1}^M {{s_l}{r_l}}$. To find the intersection, we substitute $\beta=2\beta^\prime$ in the $\beta-\beta^\prime$ plane to get
\begin{equation}
  \alpha (k - k) + k\beta \left(d + \sum\limits_{l = 1}^M {{s_l}{r_l}}  - k + \frac{{k + 1}}{2}\right) + \frac{k\beta }{2}(t - 1) = F
\end{equation}
which subsequently leads to
\begin{equation}
  \beta  = \frac{{2F}}{{k\left( {2(d + \sum\limits_{l = 1}^M {{s_l}{r_l}} ) + t - k} \right)}}.
\end{equation}
Thus, by substituting $\beta$ and $\beta^{\prime}=\beta/2$ in the expression of $\gamma_c(\mathbf{s})$,
the result follows. As for the BS-MSCR operating point, let $g=lt+q$, where $l$ is an integer and $0\le q \le t$. Note that $l=0$ corresponds to the minimum storage point i.e., $\alpha  = \frac{F}{k}$. Taking into this consideration, for a given set of values of $r_1,...,r_M$, the intersection between $\beta=\beta^\prime$ with $\beta-\beta^\prime$ plane subject to the constraint in \eqref{thrm1_2} gives us the BS-MSCR point.
\end{proof}

Note that in order to calculate BS-MSCR and BS-MBCCR operating points, we still need to find the set $\{r_l\}$. Next proposition shows that minimum bandwidth cost is indeed attained at the upper bounds. 

\begin{proposition} \label{prop:1}
  For the given values of $\rho$, $\mathbf{s}_\rho = [\underbrace {1,\dots,1}_\rho ,\underbrace {0,\dots,0}_{M - \rho }]$ and $0 \le r_l \le b_l$ for $l=1,...,M$, satisfying 
  \begin{align}
      \sum\limits_{l = 1}^\rho  {{w_l}{r_l}}  \le {\bar w_t}\sum\limits_{l = 1}^\rho  {{r_l}}
  \end{align}
  with ${{\bar w}_t} \in \left\{ {\frac{{d + t - 1}}{{d + t - k}},\frac{{2d + t - 1}}{{2d + t - k}}} \right\}$, repair bandwidth cost ${\gamma _c} \in \left\{ {{\gamma_{\textrm{BS-MSCR}}},{\gamma_{\textrm{BS-MBCCR}}}} \right\}$ is minimized at the upper bounds $[b_1,...,b_\rho]$.
\end{proposition}

\begin{proof}
In BS-MBCR case, let us assume that we happen to use $\rho$ base stations. 
In other words,$\{ {r_1},...,{r_\rho }\}$ are non-zero values satisfying the condition,
 \begin{equation}\label{w_inequality}
  \sum\limits_{l = 1}^\rho  {{w_l}{r_l}}  \le  \frac{{2d + t - 1}}{{2d - k + t}} \sum\limits_{l = 1}^\rho  {{r_l}}
\end{equation}
Note that in order for base stations to help reduce the bandwidth cost, we would need ${\gamma _c}(\mathbf{s}_\rho ) \le {\gamma _c}(\mathbf{0})$ which leads to equation \eqref{w_inequality}. Here $\mathbf{0}$ denotes the all-zero vector. 
Accordingly, inequality \eqref{w_inequality} can be rewritten as follows:
\begin{equation}\label{w_equality}
  \sum\limits_{l = 1}^\rho  {{w_l}{r_l}}  = \tau {\bar w_t}\sum\limits_{l = 1}^\rho  {{r_l}}
\end{equation}
where $0 \le \tau  \le 1$. If we reformulate ${\gamma _c}(\mathbf{s}_\rho )$  in terms of $\rho$, repair bandwidth cost can be rewritten as,
\begin{equation}\label{gamma_c_mbr}
  {\gamma _c}(\mathbf{s}_\rho ) = \frac{{{{\bar w}_t}}}{k}\frac{{2d - k + t + 2\tau \sum\limits_{l = 1}^\rho  {{r_l}} }}{{2d - k + t + 2\sum\limits_{i = 1}^\rho  {{r_i}} }}. 
\end{equation}
Note that equation \eqref{gamma_c_mbr} is inversely related to $\sum_{l = 1}^\rho  {r_l}$ (it can be easily shown by taking the derivative). In other words, in order to minimize ${\gamma  _c}(\mathbf{s}_\rho)$, it is sufficient to maximize the term $\sum_{l = 1}^\rho  {{r_l}}$. A similar proof following the same line of reasoning can be given for BS-MSCR point.
\end{proof}

The implication of the result of Proposition \ref{prop:1} is that we can simply obtain BS-MSCR and BS-MBCCR points by replacing $r_i$ with $b_i$ as,
  \begin{align}\label{optim_bs_point1}
&\left( {{\gamma_{\textrm{BS-MSCR}}},{\alpha _{\textrm{BS-MSCR}}}} \right) = \nonumber\\
& \ \ \ \ \ \ \ \ \ \ \ \ \ \ \ \ \ \left( {\frac{{F(d + \sum\limits_{l = 1}^{\rho_{\textrm{BS-MSCR}}} {{w_i}{b_i}}  + t - 1)}}{{k(d + \sum\limits_{l = 1}^{\rho_{\textrm{BS-MSCR}}} {{b_i}}  + t - k)}},\frac{F}{k}} \right)
\end{align}
\resizebox{1.02\linewidth}{!}{
\begin{minipage}{\linewidth}
\vspace{-1.5mm}
 \begin{align}\label{optim_bs_point2}
&\left( {{\gamma_{\textrm{BS-MBCCR}}},{\alpha_{\textrm{BS-MBCCR}}}} \right) =\nonumber\\
&\left( {\frac{{F(2(d + \sum\limits_{l = 1}^{\rho_{\textrm{BS-MBCCR}}} {{w_i}{b_i}} ) + t - 1)}}{{k(2(d + \sum\limits_{l = 1}^{{\rho_{\textrm{{BS-MBCCR}}}}} {{b_i}} ) + t - k)}},\frac{{F(2(d + \sum\limits_{l = 1}^{\rho_{\textrm{BS-MBCCR}}} {{b_i}} ) + t - 1)}}{{k(2(d + \sum\limits_{l = 1}^{\rho_{\textrm{BS-MBCCR}}} {{b_i}} ) + t - k)}}} \right) 
\end{align}
\end{minipage}
}

In Algorithm \ref{bs_number_cal}, we provide a linear time search algorithm to find $\rho_{\textrm{BS-MSCR}}$ and $\rho_{\textrm{BS-MBCCR}}$ i.e., the optimal number of  contributing BSs, $\rho_{min}$ for both BS-MSCR and BS-MBCCR operating points. We use $p_t$ to indicate whether the results corresponds to BS-MSCR ($p_t=1$) or BS-MBCCR $p_t=0$).  

\section{An explicit construction for Exact BS-Assisted MSCR code}

Exact repair has many advantages such as reduced maintenance compared to functional repair as indicated in previous works \cite{Shum2013}. Exact MSR and MSCR code constructions have been established for a given arbitrary code parameters \cite{Barg17}. In this section, we provide a family of regeneration codes for BS-assisted cooperative repair with parameters $\rho$ and $d=k \leq n-t$. For simplicity, we also assume $b_l=1$. Similar to previous constructions, we base our construction on maximal-distance separable (MDS) code of length $n$ with dimension $k$, consisting of symbols from $GF(p^q)$ with prime $p$ and $n \leq p^q$ for some positive integer $q$. The code presented in subsection II.A would be a special case of the simplified construction we present in this subsection.

Our construction is based on multiple MDS codes with the same $k \times n$ generator matrix\footnote{Any $k \times k$ submatrix of $\mathbf{G}$ is invertible. Thus any data collector connecting a subset of $k$ nodes would be able to reconstruct the file.} $\mathbf{G} = [\mathbf{g}_1, \mathbf{g}_2, \dots, \mathbf{g}_n]$ where $\mathbf{g}_i$ represents the $i$-th column of $\mathbf{G}$.  The file of size $F$ is divided into $k(t+\rho)$ chunks where each one of them is considered to be from $GF(p^q)$. We structure these chunks as a $(t+\rho) \times k$ massage matrix $\mathbf{M}$. Thus, we compute $\mathbf{MG}$ and make the node $j$ store $j$-th column of $\mathbf{MG}$ for $j=1,2,\dots,n$. Suppose we represent each row of $\mathbf{M}$ as $\mathbf{m}_1^T,\mathbf{m}_2^T,\dots,\mathbf{m}_{t+\rho}^T$. We have $k$ message symbols in each row represented as a vector $\mathbf{m}_i^T=[m_1 \ m_2 \ \dots \ m_k]$, each $m_i \in GF(p^m)$ encoded into the codeword $\mathbf{m}_i^T\mathbf{G}$. In such a setting, node $j$ would  store $\mathbf{m}_i^T\textbf{g}_j$, for $i=1,2,\dots,t+\rho$.


\begin{algorithm}[t!] 
\begin{algorithmic}[1]
\State \textbf{function}  \textbf{OptBSNumCal}($k,d,t,M,\mathbf{b},\mathbf{w}, p_t$)
\State $\rho_{min} \gets 0$
\For{$i=1:M$}
\State{$\overline{d} \gets d + \sum\limits_{l = 1}^{i - 1} {{w_l}{b_l}}$} \ \ , \ \ {$\overline{b} \gets d + \sum\limits_{l = 1}^{i - 1} {{b_l}}$}
\Comment{$\mathbf{b}=\{b_l\}$}
\State{${{\bar w}_t} \gets p_t\left( \frac{\overline{d}  + t - 1}{\overline{b}  + t - k} \right) + (1 - p_t)\left( {\frac{{2\overline{d} + t - 1}}{{2\overline{b} + t - k}}} \right)$}\Comment{$p_t \in \{ 0,1\}$}
    \If{$w_i > {{\bar w}_t}$}\Comment{$\mathbf{w}=\{w_l\}$}
    \State{$\textbf{break;}$}
    \EndIf
    \State{$\rho_{min} \gets \rho_{min}+1$}
\EndFor
\State{$\textbf{return}\quad\rho_{min}$}
\end{algorithmic}
\caption{Optimal Number of BSs ($\rho_{min}$)}
\label{bs_number_cal}
\end{algorithm}

Suppose that nodes with indices $j_1,j_2,\dots,j_t$ fail/depart the cell and become permanently unavailable. In the first phase of the regeneration, the $l$-th newcomer $j_l$ connects any $d=k$ remaining nodes, say $\pi_l(1), \pi_l(2), \dots, \pi_l(k)$ and downloads $\mathbf{m}_l^T\mathbf{g}_{\pi_l(1)},\mathbf{m}_l^T\mathbf{g}_{\pi_l(2)},\dots, \mathbf{m}_l^T\mathbf{g}_{\pi_l(k)}$ to successfully reconstruct $\mathbf{m}_l^T$. In the subsequent cooperative phase, newcomer $j_l$ computes $\mathbf{m}_l^T \mathbf{g}_{j_h}$ to send to the newcomer $j_h$ with $h \not = l$. So a total of $t-1$ chunks are exchanged by each newcomer in this cooperative phase. After the second phase, $l$-th newcomer $j_l$ has $t$ chunks namely $\mathbf{m}_h^T \mathbf{g}_{j_l}$ for $h=1,\dots,t$. In the final phase of the regeneration, we download the remaining $\rho$ chunks from the available $\rho$ base stations each providing one chunk of information. We first note that in this case the last two phases of the regeneration are interchangeable.
We also note that a chunk of information is of size $\frac{F}{k(t+\rho)}$ and the total number of chunks downloaded for each newcomer is $d+\rho+t-1$ achieving the minimum bandwidth possible.

\section{Conclusions and Future work}

In this work, we have considered the cooperative node repair process of distributed data storage for cellular networks in presence of base stations. Since the symbol download between local nodes and base stations may be more costly and limited, we have introduced the concept of bandwidth cost and link capacities. Furthermore, we have formulated appropriate bounds on the file size and solve the optimization problem. We conjecture that closed form expressions can be found for each operating point in the trade-off curve. Our results indicate that a better bandwidth cost-storage space trade-off may be possible in the presence of base stations. We also realize that the cost notion can be applied to storage as well and a natural extension to this work would be to characterize the bandwidth cost-storage cost trade-off. Although we have provided a  construction for the exact BS-MSCR operating point, it is valid only for specific parameter selections of the system. We leave different code constructions for general parameters and link capacity settings as a future work.

\section*{Acknowledgment}
This study is supported by TUBITAK under grant no 119E235.
\bibliographystyle{IEEEtran}
\bibliography{references}

\begin{thebibliography}{10}
\providecommand{\url}[1]{#1}
\csname url@samestyle\endcsname
\providecommand{\newblock}{\relax}
\providecommand{\bibinfo}[2]{#2}
\providecommand{\BIBentrySTDinterwordspacing}{\spaceskip=0pt\relax}
\providecommand{\BIBentryALTinterwordstretchfactor}{4}
\providecommand{\BIBentryALTinterwordspacing}{\spaceskip=\fontdimen2\font plus
\BIBentryALTinterwordstretchfactor\fontdimen3\font minus
  \fontdimen4\font\relax}
\providecommand{\BIBforeignlanguage}[2]{{%
\expandafter\ifx\csname l@#1\endcsname\relax
\typeout{** WARNING: IEEEtran.bst: No hyphenation pattern has been}%
\typeout{** loaded for the language `#1'. Using the pattern for}%
\typeout{** the default language instead.}%
\else
\language=\csname l@#1\endcsname
\fi
#2}}
\providecommand{\BIBdecl}{\relax}
\BIBdecl

\bibitem{Kermarrec2011}
A.~{Kermarrec}, N.~{Le Scouarnec}, and G.~{Straub}, ``Repairing multiple
  failures with coordinated and adaptive regenerating codes,'' in \emph{2011
  International Symposium on Networking Coding}, 2011, pp. 1--6.

\bibitem{Shum2013}
K.~W. {Shum} and Y.~{Hu}, ``Cooperative regenerating codes,'' \emph{IEEE
  Transactions on Information Theory}, vol.~59, no.~11, pp. 7229--7258, 2013.

\bibitem{sohn2019}
J.~{Sohn}, B.~{Choi}, S.~W. {Yoon}, and J.~{Moon}, ``Capacity of clustered
  distributed storage,'' \emph{IEEE Transactions on Information Theory},
  vol.~65, no.~1, pp. 81--107, 2019.

\bibitem{hou2019}
H.~{Hou}, P.~P.~C. {Lee}, K.~W. {Shum}, and Y.~{Hu}, ``Rack-aware regenerating
  codes for data centers,'' \emph{IEEE Transactions on Information Theory},
  vol.~65, no.~8, pp. 4730--4745, 2019.

\bibitem{Tebbi2019}
A.~{Tebbi}, T.~H. {Chan}, and C.~W. {Sung}, ``Multi-rack distributed data
  storage networks,'' \emph{IEEE Transactions on Information Theory}, vol.~65,
  no.~10, pp. 6072--6088, 2019.

\bibitem{zhang2021rackaware}
Z.~Zhang and L.~Zhou, ``Rack-aware regenerating codes with fewer helper
  racks,'' 2021.

\bibitem{Hou2020}
H.~{Hou}, P.~P.~C. {Lee}, and Y.~S. {Han}, ``Minimum storage rack-aware
  regenerating codes with exact repair and small sub-packetization,'' in
  \emph{2020 IEEE International Symposium on Information Theory (ISIT)}, 2020,
  pp. 554--559.

\bibitem{Ahlswede2000}
R.~{Ahlswede}, {Ning Cai}, S.~.~R. {Li}, and R.~W. {Yeung}, ``Network
  information flow,'' \emph{IEEE Transactions on Information Theory}, vol.~46,
  no.~4, pp. 1204--1216, 2000.

\bibitem{Dimakis2010}
A.~G. {Dimakis}, P.~B. {Godfrey}, Y.~{Wu}, M.~J. {Wainwright}, and
  K.~{Ramchandran}, ``Network coding for distributed storage systems,''
  \emph{IEEE Transactions on Information Theory}, vol.~56, no.~9, pp.
  4539--4551, 2010.

\bibitem{Calis2016}
G.~{Calis} and O.~O. {Koyluoglu}, ``On the maintenance of distributed storage
  systems with backup node for repair,'' in \emph{2016 International Symposium
  on Information Theory and Its Applications (ISITA)}, 2016, pp. 46--50.

\bibitem{Barg17}
M.~{Ye} and A.~{Barg}, ``Explicit constructions of high-rate mds array codes
  with optimal repair bandwidth,'' \emph{IEEE Transactions on Information
  Theory}, vol.~63, no.~4, pp. 2001--2014, 2017.

\end{thebibliography}
\end{document}